\documentclass[smallextended]{svjour3}

\usepackage{amsmath,amssymb}
\usepackage{bbm}
\usepackage{enumerate}
\usepackage{graphicx}
\usepackage[usenames,dvipsnames]{color}

\newcommand{\eps}{\varepsilon}

\smartqed
\title{Asynchronous Rumor Spreading on Random Graphs \thanks{An extended abstract of this paper has been published in the proceedings of the 24th International Symposium on Algorithms and Computation (ISAAC '13).}}

\author{K. Panagiotou \and L. Speidel}

\institute{K. Panagiotou 
			\at University of Munich, Mathematics Institute, Theresienstr. 39, 80333 Munich, Germany 
		    \and L. Speidel
			\at University of Oxford, Doctoral Training Centre in Systems Biology, Rex Richards Building, South Parks Road, Oxford OX1 3QU, UK, \email{leo.speidel@outlook.com}
}

\begin{document}

\maketitle

\begin{abstract}
We perform a thorough study of various characteristics of the asynchronous push-pull protocol for spreading a rumor on Erd\H os-R\'enyi random graphs $G_{n,p}$, for any $p>c\ln(n)/n$ with $c>1$. In particular, we provide a simple strategy for analyzing the asynchronous push-pull protocol on arbitrary graph topologies and apply this strategy to $G_{n,p}$. We prove tight bounds of logarithmic order for the total time that is needed until the information has spread to all nodes. Surprisingly, the time required by the asynchronous push-pull protocol is asymptotically almost unaffected by the average degree of the graph. Similarly tight bounds for Erd\H os-R\'enyi random graphs have previously only been obtained for the synchronous push protocol, where it has been observed that the total running time increases significantly for sparse random graphs. Finally, we quantify the robustness of the protocol with respect to transmission and node failures. Our analysis suggests that the asynchronous protocols are particularly robust with respect to these failures compared to their synchronous counterparts.

\keywords{gossip algorithms \and asynchronous rumor spreading \and push-pull protocol \and random graphs}
\end{abstract}

\section{Introduction}

Rumor spreading protocols have become fundamental mechanisms for designing efficient and fault-tolerant algorithms that disseminate information in large and complex networks. In the classical setting the algorithm that we will consider proceeds in synchronous rounds. Initially, some arbitrary node receives a piece of information. In each subsequent round, every node that knows the information transmits it to a randomly selected neighbor in the network. This operation is denoted as a \emph{push}. Moreover, every node that does not possess the information tries to learn it from a randomly selected neighbor; this operation is denoted as a \emph{pull}. Equivalently, we can say that in every round, every node contacts a randomly chosen neighbor and exchanges the information with it. 

Rumor spreading was first introduced in~\cite{Dem87}, where the problem of distributing updates consistently in replicated databases was considered. Subsequently it has found many other applications, such as the detection of failures in a distributed environment \cite{RMH98}, sampling of peers \cite{JVGKS07} and averaging in networks that consist of many sensors in a distributed fashion \cite{BABD05}. 

In this work we consider a variation of the classical push-pull algorithm that was introduced in \cite{Boyd06}. In that paper, striving for a more realistic setting, the authors modified the algorithm by dropping the assumption that all nodes are able to act in synchrony. In the \emph{asynchronous version} that we consider here, nodes do not contact other nodes simultaneously in synchronized rounds, but do so in times that arrive according to independent rate 1 Poisson processes at each node. That is, every informed node makes a push attempt and every uninformed node makes a pull attempt at a rate normalized to 1. In \cite{Boyd06} this is suggested as a possible solution if a centralized entity for facilitating time synchronization is not existent or has failed in the networks that we consider.

\subsection{Results}

In this paper we present a thorough study of various characteristics of the asynchronous push-pull algorithm. We will assume that the underlying network is an Erd\H{o}s-Renyi random graph~$G_{n,p}$, where each edge is included independently of all other edges with probability~$p$. For any~$p > c\ln(n)/{n}$, $c>1$, we show almost optimal bounds for the time that is needed until the information has spread to all nodes. We also quantify the robustness of the algorithm with respect to transmission and node failures.

Let us introduce some basic notation first. For a graph $G$ with $n$ nodes we assume that its node set is $[n]$, where $[n]=\{1, \dots, n\}$. For $v\in [n]$ we write $N_G(v)$ for the neighborhood of $v$ and $d_G(v) = |N_G(v)|$. For any $S \subseteq [n]$ we abbreviate $N_G(S) = \cup_{v \in S} N_G(v)$. Moreover, we denote by $e_G(S,R)$ the number of edges with one endpoint in each of the sets $S,R \subseteq [n]$ and abbreviate this quantity by $e_G(S)$ if $S=R$. If $G$ is clearly given from the context we may drop the subscript $G$ in our notation. Finally, we let $T(G)$ denote the (random) time that the asynchronous push-pull protocol needs to spread a rumor to all nodes in $G$, and we write $G_{n,p}$ for a random graph with $n$ vertices, where each edge is included independently with probability $p$.

We will also use the following notation. We write $H_n$ for the harmonic series $\sum_{j=1}^n j^{-1}$, $\ln(n)$ for the natural logarithm and $\log_b(n) = \ln(n)/\ln(b)$ for any $b>1$. For any $a,b\in \mathbb{R}$ we write $a \pm b$ for the interval $(a-b,a+b)$ and abbreviate $X = a \pm b$ for $X \in (a-b,a+b)$. We will not explicitly emphasize (in-)equalities that hold almost surely, i.e.\ with probability 1. Finally, w.h.p.\ abbreviates ``with high probability'' and means that an event (dependent on $n$) occurs with probability $1-o(1)$ as $n \rightarrow \infty$.
Note that when studying rumor spreading on random graphs we consider the product of two probability spaces - one for the random graph and one for rumor spreading. Usually, w.h.p.\ will correspond to sampling from the product space. Sometimes, we will average over one probability space so that w.h.p.\ refers to sampling from the other probability space. However, this will not be reflected in our notation.

Our first result addresses the performance of the algorithm on random graphs with an edge probability that is significantly above the connectivity threshold $\ln(n) /n$ for the random graph $G_{n,p}$, see also~\cite{MR1864966}. 

\begin{theorem}
\label{MainThm1}
Let $p={\alpha(n)\ln(n) }/{n}$ for some $\alpha(n)=\omega(1)$. Then w.h.p.
\[
\mathbb{E}[T(G_{n,p})]=\left( 1 \pm \sqrt{\frac{34}{\alpha(n)}} \right)H_{n-1} + O\left( \frac{\ln(n)}{n} \right).
\]
Moreover, w.h.p.\
\[
T(G_{n,p}) = \mathbb{E}[T(G_{n,p})] + O(\alpha(n)^{-1/2} \ln(n)+1).
\] 
\end{theorem}
Some remarks are in place. First of all, note that if $p \ge \ln^3(n)/n$, then the theorem states that w.h.p.\ $\mathbb{E}[T(G_{n,p})] = H_{n-1} + O(1)$, and further that w.h.p.\ $T(G_{n,p}) = \mathbb{E}[T(G_{n,p})] + O(1)$. This is tight, and indeed it is best possible, since just the time until the second node is informed has a variance of $\Omega(1)$. For almost no other rumor spreading protocol bounds that determine the total running time up to an additive constant are known (see Section \ref{sec: related work} for a discussion). On the other hand, this result might not be completely unexpected: A simple argument for large deviations from the average spreading time shows that for the complete graph $K_n$, $T(K_n) = H_{n-1} + O(1)$ w.h.p. However, for $p = \omega(\ln(n)/n)$, the edges of $G_{n,p}$ are distributed very uniformly, in the sense that between \emph{any} two sets of nodes, the number of edges is close to the expected value. So, the dynamics of the information spreading process are not affected crucially by the fact that the graph does not contain all edges.

For $p$ closer to the connectivity threshold the edges are not distributed as uniformly as in the former case and a different behavior might be expected. Indeed it has been shown that the synchronous push protocol cannot inform all nodes in a time bounded by $C \ln(n)$, with $C$ independent of $p$ (see Section \ref{sec: related work} for a discussion). Our second result, however, shows that the time required by the asynchronous push-pull protocol can be uniformly bounded independently of $p$ and is as such \emph{asymptotically} almost unaffected by the average degree of the graph.
\begin{theorem}
\label{MainThm2}
Let $p={c\ln(n)}/{n}$ for some $c>1$. Then w.h.p.
\[
\mathbb{E}[T(G_{n,p})] ~\le~ 1.5 H_{n-1} + O(\ln^{3/4}(n)).
\]
Moreover, w.h.p.\ \[
T(G_{n,p}) ~\le~ \mathbb{E}[T(G_{n,p})] + O(\ln^{3/4}(n)).
\]
\end{theorem}
This result indicates an important property of the asynchronous push-pull algorithm that has not been studied in such detail in previous works: the robustness of the algorithm with respect to the distribution of the edges and the average degree of the underlying networks. 

Theorem~\ref{MainThm2} above is a corrected version of Theorem~2 in the conference version of this paper~\cite{PS2013}. The bound on $T(G_{n,p})$ is slightly worse in the corrected version, however it is still independent of $p$. To quantify this robustness we performed numerical simulations in Section~\ref{sec:numsim}.
We find that the spreading time of asynchronous push significantly increases as $c$ decreases, while the spreading time of asynchronous pull and push-pull remain largely unaffected. In fact, in Section~\ref{sec:numsim} we will see that asynchronous pull w.h.p.\ spreads a rumor in $2\ln(n) + O(\ln^{3/4}(n))$ irrespectively of the value of $c$. We also find that push-pull is nearly twice as fast as pull, suggesting a significant contribution of push in the spreading process.


We also study the robustness of the algorithm with respect to other parameters. First, suppose that every time a node contacts some other node the connection is dropped independently of the history of the process with probability $1 - q$, where $0 < q \le 1$, before any information can be exchanged. Let $T_q(G)$ be the time until all nodes receive the information in the asynchronous push-pull protocol. Our next result quantifies the effect of the ``success probability'' $q$ on the total time.

\begin{proposition}
\label{MainThm3}
Let $0 < q \le 1$. Then, for any connected graph $G$,
$$
\mathbb{E}[T_q(G)] = \frac{1}{q}~\mathbb{E}[T(G)].
$$
Moreover, w.h.p.\ 
\begin{alignat*}{2}
	T_q(G_{n,p}) &= \frac{1}{q}\mathbb{E}[T(G_{n,p})] + O(\alpha(n)^{-1/2}\ln(n)+1) & \quad ~\textrm{if}~ & p = \alpha(n)\ln(n)/n, \\
	T_q(G_{n,p}) &= \frac{1}{q}\mathbb{E}[T(G_{n,p})] + O(\ln^{3/4}(n)) & \quad ~\textrm{if}~ & p = c\ln(n)/n, c>1.
\end{alignat*}
\end{proposition}
Note that $1/q$ is exactly the expected number of connection attempts that have to be made until the connection is not dropped for the first time. Therefore, this result also demonstrates the robustness and the adaptivity of the asynchronous push-pull algorithm that essentially slows down at the least possible rate. 

On $G_{n,p}$ with $p = \alpha(n)\ln(n)/n$, for some $\alpha(n) = \omega(1)$, it is known that under the \emph{synchronous} push protocol, all nodes are informed in $(1/q+1/\ln(1+q)) \ln(n)$ rounds w.h.p., where lower order terms are omitted~\cite{FHP10}. We can show analogous results to Theorem~\ref{MainThm1} and Proposition~\ref{MainThm3} for the \emph{asynchronous} push protocol, and obtain that the time to inform all nodes is given by $2/q \ln(n)$ w.h.p., again omitting lower order terms. This is faster than its synchronous counter part and, in particular, the speed difference increases with larger $q$.

Finally, we study the robustness with respect to node failures. Suppose that in the given network a random subset $B$ of nodes is declared ``faulty'', in the sense that even if they receive the rumor, they will neither perform any push operation, nor will they respond to any pull request. Let $T_B(G)$ denote the time until the information has spread to all nodes. Our next result states that $T_B(G_{n,p})$ is asymptotically equal to $T(G_{n,p})$, provided that $B$ is not linear in $n$ and the initially informed node, which is fixed and cannot be chosen after the graph has been sampled, does not fault.
\begin{proposition}
\label{MainThm4}
Let $p$ be as in Theorem~\ref{MainThm1} or Theorem~\ref{MainThm2}, and suppose that~$B = o(n)$. Then w.h.p.
\[
	T_B(G_{n,p}) = (1 + o(1))\mathbb{E}[T(G_{n,p})].
\]
\end{proposition}

\subsection{Related Work}
\label{sec: related work}

There are many theoretical studies that are concerned with the performance of the \emph{synchronous} push-pull algorithm~\cite{CLP10-2,CLP10,DFF11,FPS12,Gia11,KSSV00}. For example, the performance of rumor spreading on general graph topologies was made explicit in \cite{CLP10-2,CLP10,Gia11}, where the number of rounds necessary to spread a rumor was related to the conductance of the graph. In particular, the upper bound $O(\varphi^{-1}\ln(n))$ was shown, where $\varphi$ is the conductance of the graph.

It is also known that push-pull is efficient on many classes of random graphs. For example, on classical preferential attachment graphs \cite{BA99} it was shown in \cite{DFF11} that w.h.p.\ a rumor spreads in $\Theta(\ln(n))$ rounds. Moreover, the performance of the push-pull algorithm on classes of random graphs with a degree sequence that is a power law was studied in \cite{FPS12}. In particular, they showed that if the degree sequence has unbounded variance, then the number of rounds until the information has reached almost all nodes is reduced to $O(\ln\ln(n))$, while in all other cases it remains $O(\ln(n))$. 

For the synchronous push protocol there exist very accurate bounds of the spreading time~\cite{FHP10,FG85,Pit87}. In \cite{Pit87} it was shown that a rumor spreads in $\log_2(n)+\ln(n)+O(1)$ rounds w.h.p.\ on the complete graph. In \cite{FHP10} it was shown that for $G_{n,p}$ with $p=\omega(\ln(n)/n)$ a rumor spreads w.h.p.\ in $\log_2(n)+\ln(n)+o(\ln(n))$ rounds. Furthermore, it was recently shown that the spreading time increases significantly when approaching the connectivity threshold more closely, namely that for $G_{n,p}$ with $p=c\ln(n)/n$, $c>1$, the spreading time is w.h.p. $\log_2(n)+\gamma(c)\ln(n)+o(\ln(n))$, where $\gamma(c)=c\ln(c/(c-1))$~\cite{PPSS}. This means in particular that the spreading time cannot be bounded by logarithmic time independent of the edge probability $p$; this is in contrast to the asynchronous push-pull algorithm, which is more robust with respect to variations in the average degree, cf.\ Theorem~\ref{MainThm2}. On random regular graphs with degree $d \ge 3$, the spreading time has been shown to equal $(1+2/d)\gamma(d)\ln(n)+o((\ln\ln(n))^2)$ w.h.p.~\cite{Panagiotou2013}.

In addition, the effect of transmission failures that occur independently for each contact at some constant rate $1-q$ was investigated in \cite{FHP10} for the synchronous push protocol on dense random graphs. There, it was shown that the total time increases under-proportionally, namely that it is w.h.p. equal to $(1/q+1/\ln(1+q)) \ln n + o(\ln n)$~\cite{FHP10}. While the total time increases proportionally, i.e., by a factor of $1/q$, in the asynchronous case, the asynchronous protocol is still faster and the speed difference increases with larger $q$.

For the asynchronous push-pull protocol there exists much less literature~\cite{DFF12,FPS12}, mostly devoted to models for scale-free networks. In \cite{DFF12} it was shown that on preferential attachment graphs a rumor needs a time of $O(\sqrt{\ln(n)})$ w.h.p.\ to spread to almost all nodes. On power-law Chung-Lu random graphs \cite{CL03} (for power-law exponent $2<\beta<3$) it was shown in \cite{FPS12} that a rumor initially located within the giant component spreads w.h.p.\ even in constant time to almost all nodes. Related to the asynchronous push-pull protocol is first-passage percolation, which, on regular graphs, is an equivalent process. There, it has been shown that the running time on the hypercube and the complete graph is $\Theta(\ln(n))$~\cite{Bollobas1997book,Fill1993,Janson1999}. In a recent study, the ratio of the spreading time of the synchronous and asynchronous push-pull protocol was bounded by $\Omega(1/\ln(n))$ from below and by $O(n^{2/3})$ from above~\cite{Acan_arXiv}. In particular, examples of graphs in which the asynchronous version spreads the rumor in logarithmic time and the synchronous version needs polynomial time were given.


\section{Preliminaries}
\label{sec:prel}

We will exploit the following properties of $G_{n,p}$. Let us begin with the case $p=\alpha(n)\ln(n)/n$, $\alpha(n)=\omega(1)$. Here we will use the fact that the edges are w.h.p.\ distributed very ``uniformly'', in the sense that between any $S \subseteq [n]$ and its complement the number $e(S, [n] \setminus S)$ is very close to its expected value $p(n-\vert S \vert) \vert S \vert$.
\begin{lemma}[see e.g.\ Lemma IV.3 in \cite{FHP10}]
\label{lemma:dense}
Let $p = { \alpha(n) \ln(n)}/{n}$, where $\alpha(n)  = \omega(1)$. Then, w.h.p. $G_{n,p}$ is such that for all $S \subseteq [n]$
\begin{equation}
\label{property}
e(S,[n] \setminus S)= \left(1 \pm \sqrt{\frac{8}{\alpha(n)}} \right ) \left(n-\vert S \vert \right)\vert S \vert p.
\end{equation}
\end{lemma}
For the case that $p={c\ln(n)}/{n}$, where $c>1$, the conclusion of the previous lemma is not true; actually, there are significant fluctuations in the quantity $e(S, [n]\setminus S)$ for different sets $S$ of the same size. Instead, we will exploit the following properties. First, we use that w.h.p.\ the degree of any node is of logarithmic order.
\begin{lemma}[see e.g.~\cite{MR1864966}]
\label{lemma:deg(v)}
Let $c > 1$ and $p={c\ln(n)}/{n}$. Then there are constants $C'$ and $C$ depending on $c$ with $C' > C > 0$, such that w.h.p.\ for any $v\in [n]$ we have
$$C \ln(n) \le d(v) \le C' \ln(n).$$
\end{lemma}
The next properties give us (coarse) information about $e(S)$ and $e(S,[n] \setminus S)$ for $S \subseteq [n]$ in the spirit of Lemma~\ref{lemma:dense}. In particular, if $|S|$ is not too large, then w.h.p. the average degree of the subgraph induced by $S$ is exponentially smaller compared to the average degree of $G_{n,p}$.
\begin{lemma}
\label{lemma:e(S,[n]-S)}
Let $c > 1$ and $p={c\ln(n)}/{n}$. Then w.h.p.\ 
\begin{alignat*}{2}
	e(S,[n] \setminus S) &
 		= \Theta(\vert S \vert \ln(n) ) &
 	&	\quad \textrm{for any } S \subseteq [n] \textrm{ with } \vert S \vert \leq {n}/{2}, \\  
	e(S) &
		\le \vert S \vert \ln\ln(n) &
	&	\quad \textrm{for any } S \subseteq [n] \textrm{ with } \vert S \vert \leq {n}/\ln(n).
\end{alignat*}
\end{lemma}
The first statement can be found e.g.\ in~\cite{CF07}, Property 3, while the second statement is shown within the proof of Property 3 that is provided in the Appendix of the same paper. For sets $S$ with size exceeding $\ln(n)$ we will require stronger bounds. The next statement addresses \emph{connected} sets $S$, i.e., where the subgraph of $G_{n,p}$ induced by $S$ is connected.
\begin{lemma}
\label{lemma:S connected}
Let $c > 1$ and $p = {c\ln(n)}/{n}$. Then w.h.p.\ the graph is such that all connected sets $S \subseteq [n]$ with $\ln(n) \leq \vert S \vert \leq {n}/2$ fulfill
\[
e(S,[n] \setminus S) = (1 \pm \eps(n)) \, \vert S \vert (n-\vert S \vert)p,
\]
where $\eps(n) = \left(\frac{24\ln\ln(n)}{\ln(n)}\right)^{1/2}$.
\end{lemma}
\begin{proof}
Let $|S| = s$, where $s \in \left[\ln(n),{n}/2 \right]$. We write $\eps = \eps(n)$ for short. In the sequel we show that
\begin{equation}
\label{eq:tmpfixeds}
	\Pr\big[\exists \text{ connected } S: ~ |S| = s \text{ and } e(S,[n] \setminus S) \not\in (1 \pm \eps) \, s(n-s)p\big] \le 2ne^{-2s\ln\ln(n)}.
\end{equation}
By applying a union bound for all $s$ in the considered range the statement follows.

In order to estimate the probability in~\eqref{eq:tmpfixeds} note first that there exist $\binom{n}{s}$ sets of size $s$. If the random graph restricted to $S$ is connected then we can find a spanning tree within this subgraph. By Cayley's formula there exist $s^{s-2}$ distinct trees with $s$ nodes (proofs can be found e.g. in \cite{AZ10}, Chapter 30), and so $s^{s-2}p^{s-1}$ is a upper bound for the probability that $G_{n,p}$ contains any one of them. Moreover, in $G_{n,p}$ the edges with both endpoints in $S$ are independent from the edges with at most one endpoint in $S$. Therefore, the probability in~\eqref{eq:tmpfixeds} can be bounded by
\[
\binom{n}{s} s^{s-2}p^{s-1} \cdot \Pr\big[\vert e(S,[n] \setminus S) - s(n-s)p \vert > \eps s(n-s)p\big],
\]
where $S$ denotes any set of size $s$.
The quantity $e(S,[n] \setminus S)$ in $G_{n,p}$ is binomially distributed with parameters $s(n-s)$ and $p$. We will apply the following version of the Chernoff bound (see e.g.~\cite{MR1864966}). For a binomially distributed random variable $X$ and any $t>0$, 
\begin{equation}
\label{Chernoff}
\Pr[\vert X-\mathbb{E}[X] \vert > t] \le 2\exp\left\{\frac{-t^2}{2\mathbb{E}[X]+2t/{3}} \right\}.
\end{equation}
In addition $\binom{n}{s} \leq \left( \frac{en}{s} \right)^s$. By putting everything together the expression in~\eqref{eq:tmpfixeds} can be bounded for sufficiently large $n$ by
\[
2 \left( \frac{en}{s} \right)^s s^{s-2}p^{s-1}\, \exp\left\{-\frac{\eps^2}{3} s(n-s)p \right\}
\]
that with our assumptions equals
\begin{equation}
\label{proof:S connected}
\frac{2}{cs^2} \frac{n}{\ln(n)}
\exp \left\{ s \left( \ln(ec \ln(n))- \frac{c\eps^2}{3} \frac{n-s}{n} \ln(n) \right) \right\}.
\end{equation}
For sufficiently large $n$ we may bound $ec \le \ln(n)$ and $n-s \ge n/2$. Thus, by using the definition $\eps = (\frac{24\ln\ln(n)}{\ln(n)})^{1/2}$ the previous expression is for large $n$ at most
\[
2n \, \exp \left\{ s \left( 2\ln\ln(n)- \frac{\eps^2}{6} \ln(n) \right) \right\}
\le 2n e^{-2s\ln\ln(n)},
\]
as claimed. \qed
\end{proof}

Finally we consider the neighborhood of nodes. The following lemma states that w.h.p.\ most neighbors of some $v \in [n]$ have a degree very close to the expected degree\footnote{This observation was also made in~\cite{PPSS}.}.
\begin{lemma}
\label{lemma:N(v)}
Let $c > 1$ and $p = {c\ln(n)}/{n}$. For $v \in [n]$ set
\[
N'(v) = \left\{ w \in N(v): d(w) = pn \pm \ln^{3/4}(n) \right\}.
\]
Then w.h.p.\ for every $v \in [n]$ we have $\vert N(v) \setminus N'(v) \vert \leq \ln^{3/4}(n)$.
\end{lemma}

\begin{proof}
Throughout the proof we assume that for any $u \in [n]$ there exists a constant $d$ such that $d(u) \leq d \ln(n)$, which is justified by Lemma \ref{lemma:deg(v)}. We fix a $u \in [n]$ and let $\tilde{N} := N(u) \backslash N'(u)$. To prove the statement we show that $\Pr[|\tilde{N}| > \ln^{3/4}(n)] = o(1/n)$ so that we can apply a union bound over all vertices in $u \in [n]$.

First, we consider the internal edges of the neighbourhood of $u$. Note that there are at most $(d \ln(n))^4$ possible pairs of edges inside $N(u)$, so that we get

\begin{equation}
\label{eq:Pr[>2]}
\Pr[e(N(u) \geq 2 ] \leq (d \ln(n))^4 p^2 = n^{-2+o(1)}.
\end{equation}
Next we consider the edges between $N(u)$ and $V_0 := [n] \setminus ( N(u) \cup \{u\} )$. For any $v \in [n]$ we consider the random number $X_v := | N(v) \cap V_0 |$ that describes the number of neighbours of $v$ lying in $V_0$. $X_v$ is binomially distributed with parameters $|V_0|$ and $p$. Thus we can use the version of the Chernoff bounds described in \eqref{Chernoff} (within the proof of Lemma \ref{lemma:S connected}) to obtain
\[
\Pr[|X_v-np|>\ln^{3/4}(n)-2] = e^{-\Omega(\ln^{1/2}(n))}.
\]
We define $B:=|\{v: |X_v-np|>\ln^{3/4}(n)-2\}|$. 
Since the random numbers $X_v$ are independent by construction, we get
\begin{align}
\label{eq:Pr[B]}
\Pr[B>\ln^{3/4}] 
&\leq d(u) \binom{d(u)}{\ulcorner \ln^{3/4}(n) \urcorner} \exp\{\Omega(\ln^{1/2}(n)) \ulcorner \ln^{3/4}(n) \urcorner \} \notag\\
&= e^{-\Omega(\ln^{5/4}(n))}.
\end{align}
Combinung \eqref{eq:Pr[>2]} and \eqref{eq:Pr[B]} yields
\[
\Pr[(e(N(u) \geq 2) \cup (B>\ln^{3/4})] \leq n^{-2+o(1)}=o\left(\frac{1}{n} \right).
\]
Finally note that $e(N(u)) < 2$ implies $X_v \leq d(v) \leq X_v+1$. This means that $\tilde{N}>\ln^{3/4}(n)$ implies $(e(N(u))\geq 2) \cup (B>\ln^{3/4}(n))$ which concludes the proof.\qed
\end{proof}

\section{Analysis of the Protocol}
\label{sec:aop}

In this section we describe a simple strategy for analyzing the spreading time on arbitrary (connected) graphs. It should be noted that this strategy only applies when contact times are governed by Poisson processes. We will apply this strategy to $G_{n,p}$ in the section hereafter. Let $G$ be any connected graph with $n$ nodes. Towards studying the distribution of $T(G)$ we divide the rumor spreading process into $n$ \emph{states}, where state $1 \leq j \leq n$ stands for the situation that $j$ nodes are informed. In what follows we denote the set of informed nodes in state $j$ by $I_j = I_j(G)$ and the set of uninformed nodes in state $j$ by $U_j = U_j(G) = [n]\setminus I_j(G)$. We denote by $t_j=t_j(G)$ the (random) time that the protocol needs to move from state $j$ to state $j+1$, where $1 \leq j < n$. Clearly $T(G)=\sum_{j=1}^{n-1} t_j$. 

Assume that $1 \leq j < n$ nodes are informed. We provide a general lemma that determines the distribution of $t_j$ for an arbitrary set of informed (and correspondingly uninformed) nodes.
\begin{lemma}
\label{Masterlemma}
Let $1 \le j < n$. Then $t_j$ is exponentially distributed with parameter 
\[
Q_j :=\sum_{v \in I_j} \frac{\vert N(v) \cap U_j \vert}{d(v)} + \sum_{w \in U_j} \frac{\vert N(w) \cap I_j \vert}{d(w)}.
\]
Moreover, conditional on $I_j$ the time $t_j$ is independent of $t_1, \ldots t_{j-1}$.
\end{lemma}

\begin{proof}
We assume that the times~$t_1, \ldots, t_{j-1}$ and~$I_j$ are known and determine the distribution of~$t_j$ by applying standard tools of probability theory~\cite{feller_book}. The probability that~$v \in I_j$ informs a uninformed node in a push attempt is~${\vert N(v) \cap U_j \vert}/{d(v)}$. Similarly, the probability that~$w \in U_j$ is informed in a pull attempt is~${\vert N(w) \cap I_j \vert}/{d(w)}$.  Therefore, the probability that a uninformed node is informed in a push or pull attempt is
\[
q_j = \frac{1}{n}~ \left(\sum_{v \in I_j} \frac{\vert N(v) \cap U_j \vert}{d(v)} + \sum_{w \in U_j} \frac{\vert N(w) \cap I_j \vert}{d(w)}\right).
\]
Since all nodes are equipped with rate~$1$ Poisson processes, the time between two consecutive push or pull attempts is exponentially distributed with parameter $n$. It follows that $t_j$ is exponentially distributed with parameter $Q_j = nq_j$, as claimed.\qed
\end{proof}


\section{The Expected Spreading Time on Random Graphs}
\label{sec:pthm1thm2}

Here we apply the results of the previous section to $G_{n,p}$. In particular, we determine the expected value of $T(G_{n,p})$ thus proving the first statements of Theorem \ref{MainThm1} and \ref{MainThm2} respectively.

\subsection{The case $p=\omega({\ln(n) }/{n})$}

Our goal is to compute the expectation of $T(G_{n,p})$ for $p=\alpha(n)\ln(n)/n$, $\alpha(n) = \omega(1)$. We will actually show a stronger result, namely that the conclusion of Theorem \ref{MainThm1} remains valid even when we replace $G_{n,p}$ with \textit{any} graph $G$ that satisfies the conclusion of Lemma \ref{lemma:dense}, with $p$ replaced by $\alpha(n)\ln(n)/n$. To this end, we first specify the distribution of $t_j(G)$ by applying Lemma \ref{Masterlemma}.

\begin{corollary}
\label{dense:t_j}
Let $G$ be any graph satisfying the conclusion of Lemma \ref{lemma:dense}. If $n$ is sufficiently large, then for any $I_j$, the distribution of $t_j(G)$ conditional on $I_j$ is an exponential distribution with parameter 
\begin{equation}
\label{eq:parboundsdense}
\left(1 \pm \sqrt{\frac{33}{\alpha(n)}} \right) \frac{2j(n-j)}{n-1}.
\end{equation}
\end{corollary}

\begin{proof}
Let~$I_j$ be any (connected) set with~$j$ nodes. For any node~$v$ in~$G$, by using~\eqref{property} with~$S=\{v\}$ we infer that~$d(v)= \big(1 \pm \sqrt{8/\alpha(n)} \big)(n-1)p$. Moreover, we have~$\sum_{v \in I_j} \vert N(v) \cap U_j \vert = e(I_j,U_j).
$
By applying \eqref{property} for a second time we obtain that
\[
\sum_{v \in I_j} \frac{\vert N(v) \cap U_j \vert}{d(v)}
= \frac{1 \pm \sqrt{8/\alpha(n)}}{1 \mp \sqrt{8/\alpha(n)}} ~ \frac{j(n-j)}{n-1}
\]
and simple algebraic transformations imply that we can bound this expression for sufficiently large~$n$ by~$\big( 1 \pm \sqrt{33/\alpha(n)} \big) j(n-j)/(n-1)$. We repeat the above calculation with $I_j$ and $U_j$ interchanged. Finally, we plug this into~$Q_j$ of Lemma \ref{Masterlemma} and the statement is shown. \qed
\end{proof}

We have now everything together to calculate the expectation of $T(G)$. Note that the first statement of Theorem \ref{MainThm1} follows immediately from Lemma \ref{dense: expectation}.

\begin{lemma}
\label{dense: expectation}
Let $G$ be any graph satisfying the conclusion of Lemma \ref{lemma:dense}. Then
\[
\mathbb{E}[T(G)]=\left( 1 \pm \sqrt{\frac{34}{\alpha(n)}} \right)H_{n-1} + O\left(\frac{\ln(n)}{n} \right).
\]
\end{lemma}

\begin{proof}
Using Corollary \ref{dense:t_j} we obtain that $\mathbb{E}[T(G)]$ equals
\[
\sum_{j=1}^{n-1} \mathbb{E}[t_j] 
=\sum_{j=1}^{n-1} \mathbb{E}[ \mathbb{E}[ t_j ~\vert~ I_j ] ] 
= \left(1 \pm \sqrt{\frac{33}{\alpha(n)}}\right)^{-1} \sum_{j=1}^{n-1} \frac{n-1}{2j(n-j)}.
\]
Additionally, by making use of the bound $(1 \pm \sqrt{{33}/{\alpha(n)}})^{-1} = 1 \pm \sqrt{{34}/{\alpha(n)}}$ for sufficiently large $n$ and using the identity $(j(n-j))^{-1} = (jn)^{-1} + (n(n-j))^{-1}$ we obtain that
\[
\mathbb{E}[T(G)]= \left( 1 \pm \sqrt{\frac{34}{\alpha(n)}} \right) \left(H_{n-1} - \frac{1}{n} H_{n-1} \right).
\]
Together with $H_n=\ln(n)+O(1)$ we finally arrive at the claimed bound. \qed
\end{proof}

\subsection{The case $p={c\ln(n)}/{n}$, $c>1$}

We will again show a stronger result. This time we will prove that the conclusion of Theorem \ref{MainThm2} remains true even when we replace $G_{n,p}$ with any graph $G$ that satisfies the properties described in the Lemmas \ref{lemma:deg(v)} to \ref{lemma:N(v)}, with $p$ replaced by $c\ln(n)/n$. We begin with determining the distribution of $t_j(G)$; here the bounds are not as tight as in Corollary~\ref{dense:t_j} for all $j$, but they will suffice for our purposes.
\begin{corollary}
\label{sparse:tj}
Let $c>1$ and let $G$ be any graph satisfying the conclusions of Lemma \ref{lemma:deg(v)} -- \ref{lemma:N(v)}. Then there are constants $C'$ and $C$ depending on $c$ with $C'>C>0$, such that for any $I_j$ the distribution of $t_j(G)$ conditional on $I_j$ is an exponential distribution with parameter
\begin{enumerate}[i)]
	\item in $(C j,C'j)$ for $1\le j \le \ln(n)$,
	\item in $\big(1 \pm C\ln^{-1/4}(n)\big)2j$ for $\ln(n) \le j\le n/\ln^3(n)$,
	\item in $(C\min\{j, n-j\},C'\min\{j, n-j\})$ for $n/\ln^3(n) \le j \le n-n/\ln^3(n)$,
	\item exceeding $\big(1 - C \ln^{-1/4}(n) \big) (n-j)$ for $\ln(n) \le n-j \le n/\ln^3(n)$,
	\item in $(C(n-j),C'(n-j))$ for $1\le n-j \le \ln(n)$,
\end{enumerate}
\end{corollary}

\begin{proof}
Let $I_j$ be some connected subset of $[n]$ with $j$ elements. We apply Lemma \ref{Masterlemma} and determine $Q_j$.
Let us begin with the case $1 \le j \le \ln(n)$. By using Lemma~\ref{lemma:deg(v)} we may assume that $d(v) = \Theta(\ln(n))$. Additionally, by Lemma~\ref{lemma:e(S,[n]-S)} we may assume that $e(I_j, U_j) = \Theta(j \ln(n))$. These two properties together imply that $Q_j$ equals 
\[
\Theta\left(\sum_{v \in I_j} \frac{\vert N(v) \cap U_j \vert}{\ln(n)} + \sum_{w \in U_j} \frac{\vert N(w) \cap I_j \vert}{\ln(n)} \right) 
= \Theta\left( \frac{e(I_j,U_j)}{\ln(n)}\right)=\Theta(j).
\]
The claim for $1 \le n-j \le \ln(n)$ follows analogously by interchanging the roles of $I_j$ and $U_j$.

Next we consider the case $\ln(n) \le j \le n/\ln^3(n)$, where we bound $Q_j$ more accurately. We begin with the first sum in the expression for $Q_j$. Using that $e(I_j, U_j) = \Theta(j \ln(n))$ we obtain 
\[
\sum_{v \in I_j} \frac{\vert N(v) \cap U_j \vert}{d(v)} = j - \sum_{v \in I_j} \frac{\vert N(v) \cap I_j \vert}{d(v)}=j - \Theta\left(\frac{e(I_j)}{\ln(n)}\right).
\]
To estimate $e(I_j)$ we make use of the second statement in Lemma \ref{lemma:e(S,[n]-S)}, namely that for any $S \subseteq [n] $ with $\vert S \vert \leq n / \ln(n)$ we have that $e(S)\le\vert S \vert \ln \ln(n)$. This implies that
\begin{equation}
\label{eq:tmpQ1}
j - \Theta\left(\frac{e(I_j)}{\ln(n)}\right)=\left(1-O\left(\frac{\ln\ln(n)}{\ln(n)}\right)\right)j.
\end{equation}
We move on to the second sum of $Q_j$. We split this sum into three parts, namely into $\{w \in N'(I_j)\}$, $\{w \in N(I_j) \setminus N'(I_j)\}$ and $\{w \in U_j \setminus N(I_j)\}$, where we define $N'(S)$ to be the set containing all nodes that are outside of $S$ but belong to the neighborhood of $S$ and have a degree in $c\ln(n) \pm \ln^{3/4}(n)$, where $c\ln(n) = pn$. The reason for doing so is that Lemma \ref{lemma:N(v)} assures that among the adjacent nodes of any node $v \in [n]$ only a sub-logarithmic number of nodes have a degree outside of $c\ln(n) \pm \ln^{3/4}(n)$.

Among the three sums, the latter sum equals zero so that
\[
\sum_{w \in U_j} \frac{\vert N(w) \cap I_j \vert}{d(w)} 
 = \sum_{w \in N'(I_j)} \frac{\vert N(w) \cap I_j \vert}{d(w)} + \sum_{w \in N(I_j) \setminus N'(I_j)} \frac{\vert N(w) \cap I_j \vert}{d(w)}.
\]
Using Lemma \ref{lemma:deg(v)} and the definition of $N'(I_j)$ we infer that this is
\begin{equation}
\label{proof:lem1.2:eq}
\frac{e(I_j,N'(I_j))}{pn\pm \ln^{3/4}(n)}
	+ \frac{e(I_j, N(I_j) \setminus N'(I_j))}{\Theta(\ln(n))}.
\end{equation}
Thus we need to consider $e(I_j,N'(I_j))$ and $e(I_j, N(I_j) \setminus N'(I_j))$. By applying Lemma \ref{lemma:N(v)} we have that 
\begin{equation}
\label{eq:e(Ij,N(Ij)-N'(Ij))}
e(I_j,N(I_j) \setminus N'(I_j) )
\leq \sum_{v \in I_j} \vert N(v) \setminus N'(v) \vert
\leq j \ln^{3/4}(n). 
\end{equation}
To estimate $e(I_j,N'(I_j))$ we make use of the property in the conclusion of Lemma \ref{lemma:S connected}. Note that $I_j$ is necessarily a connected set in $G$, as any node gets the rumor from one of its neighbors. Thus Lemma \ref{lemma:S connected} guarantees that 
\begin{equation}
e(I_j,[n] \setminus I_j) = (1 \pm \eps(n)) \, \vert I_j \vert (n-\vert I_j \vert)p,
\end{equation} 
Moreover, using \eqref{eq:e(Ij,N(Ij)-N'(Ij))} we obtain that
\[
e(I_j,N'(I_j))
= \big(1 - O(\ln^{-1/4}(n))\big)j(n-j)p.
\]
Together with \eqref{eq:tmpQ1}, \eqref{proof:lem1.2:eq} and \eqref{eq:e(Ij,N(Ij)-N'(Ij))} this implies
$
Q_j=\big(1-O(\ln^{-1/4}(n))\big)2j
$.

Next we consider the case $n/\ln^3(n) \le j \le n-n/\ln^3(n)$. Again using that $d(v) = \Theta(\ln(n))$, and that $e(I_j, U_j) = \Theta(\min\{j, n-j\} \ln(n))$ we obtain that 
\[
Q_j = \Theta\left(\frac{e(I_j,U_j)}{\ln(n)}\right)=\Theta(\min\{j, n-j\}).
\]
This shows the statement also for the case $n/\ln^3(n) \le j \le n-n/\ln^3(n)$. 

Finally, for $\ln(n) \le n-j \le n/\ln^3(n)$, we use that $e(I_j,U_j) = \Theta((n-j)\ln(n))$ and obtain
\begin{align*}
\sum_{v \in I_j} \frac{\vert N(v) \cap U_j \vert}{d(v)} + \sum_{w \in U_j} \frac{\vert N(w) \cap I_j \vert}{d(w)} 
& \ge j - \sum_{w \in U_j} \frac{\vert N(w) \cap U_j \vert}{d(w)} \\
& =(n - j) - \Theta\left(\frac{e(U_j)}{\ln(n)}\right).
\end{align*}
To estimate $e(U_j)$ we make use of the second statement in Lemma \ref{lemma:e(S,[n]-S)}, namely that for any $S \subseteq [n] $ with $\vert S \vert \leq n / \ln(n)$ we have that $e(S)\le \vert S \vert \ln \ln(n) \le  \vert S \vert \ln^{3/4}(n)$. This implies that
\begin{equation}
(n-j) - \Theta\left(\frac{e(U_j)}{\ln(n)}\right)=\left(1-O\left(\ln^{-1/4}(n)\right)\right)(n-j).
\end{equation}
\qed
\end{proof}

With the above lemma at hand we are able to compute the expectation of $T(G)$. Note that the first statement of Theorem \ref{MainThm2} follows immediately from Lemma~\ref{sparse: expectation}.
\begin{lemma}
\label{sparse: expectation}
Let $c>1$ and let $G$ be any graph satisfying the conclusions of Lemmas \ref{lemma:deg(v)} -- \ref{lemma:N(v)}. Then
\[
\mathbb{E}[T(G)] \le 1.5 H_{n-1} + O(\ln^{3/4}(n)).
\]
\end{lemma}
\begin{proof}
Recall that $T(G) = \sum_{j = 1}^n t_j$. We will apply Lemma \ref{sparse:tj} several times. First of all, if $1\le j \le \ln(n)$ or $1\le n-j \le \ln(n)$ then Lemma \ref{sparse:tj} guarantees the existence of a $C = C(c) > 0$ such that
\[
	\mathbb{E}[t_j] = \mathbb{E}[\mathbb{E}[ t_j ~\vert~ I_j]] \le (C\min\{j, n-j\})^{-1}.
\]
Using the bound $H_n = \ln(n) + O(1)$ we obtain
\begin{equation}
\label{eq:phase1+5}
0
\le
\mathbb{E} \left[\sum_{j=1}^{\ln(n)} t_j + \sum_{j=n - \ln(n)}^{n} t_j\right] 
\le \frac{2}{C} \sum_{j=1}^{\ln(n)} \frac1j = \frac2CH_{\ln(n)} = O(\ln\ln(n)).
\end{equation}
Similarly, by applying Lemma \ref{sparse:tj} for $\ln(n) \le j \le n/\ln^3(n)$ and $\ln(n)\le n-j \le n/\ln^3(n)$ we get with the abbreviation $\eps(n) = C\ln^{-1/4}(n)$ for sufficiently large $n$
\begin{equation}
\label{eq:phase2+4}
\mathbb{E} \left[\sum_{j=\ln(n)}^{n/\ln^3(n)} t_j + \sum_{j=n-n/\ln^3(n)}^{n-\ln(n)} t_j\right] 
\le (1 + \eps(n)) \cdot \sum_{j=\ln(n)}^{n/\ln^3(n)} \frac{3}{2j} 
= (1 \pm \eps(n)) \frac{3}{2} H_{n-1}.
\end{equation}
Finally, by applying Lemma \ref{sparse:tj} to all remaining $j$ we get that 
\[
0
\le
\mathbb{E} \left[\sum_{j=n/\ln^3(n)}^{n - n/\ln^3(n)} t_j \right]
\le \frac2C \sum_{j=n/\ln^3(n)}^{n/2} \frac{1}{j} 
= \frac{2}{C}(H_{n/2} - H_{n/\ln^3(n)-1}).
\]
Using again that $H_n=\ln(n)+O(1)$ the last expression simplifies to $O(\ln\ln(n))$. By summing this up together with~\eqref{eq:phase1+5} and~\eqref{eq:phase2+4} and using the fact $T(G) = \sum_{j=1}^n t_j$ we arrive at the claimed bound.  \qed
\end{proof}


\section{The Actual Spreading Time on Random Graphs}
\label{sec:pt3t4}

In this section we will complete the proofs of Theorems \ref{MainThm1} and \ref{MainThm2} by showing that the time for a rumor to spread to all nodes is concentrated around the expected value. As in the previous section we prove a stronger statement in replacing $G_{n,p}$ by any graph $G$ that satisfies the assumptions made in Corollary \ref{dense:t_j} or Corollary~\ref{sparse:tj}.

To begin with, we prove two lemmas motivated by the following fact derived in Corollaries \ref{dense:t_j} and \ref{sparse:tj}. There we studied the distribution of $t_j$ conditional on $I_j$ and found that it is an exponential distribution with a parameter that can be bounded uniformly for $I_j$. This means in particular that it is possible to find deterministic sequences~$f(n,j),g(n,j)$ such that $f(n,j) \leq \mathbb{E}[t_j ~\vert~ I_j]^{-1} \leq g(n,j)$. We show that these facts are sufficient to prove that the sequence of times $\{t_j\}_{j=1}^{n-1}$ can be stochastically bounded by sequences of independent random variables, for which we will later derive large deviation estimates.
\begin{lemma}
\label{stochbound}
Let $f(n,j)$ and $g(n,j)$ be deterministic sequences such that for $1 \le j < n$
$$f(n,j) \leq \mathbb{E}[t_j ~\vert~ I_j]^{-1} \leq g(n,j).$$
Moreover, let $\{t_j^+\}_{j=1}^{n-1}$, $\{t_j^-\}_{j=1}^{n-1}$ be sequences of independent random variables, where $t_j^+$ is exponentially distributed with parameter $f(n,j)$ and $t_j^-$ is exponentially distributed with parameter $g(n,j)$. Then $\{t_j^+\}_{j=1}^{n-1}$ stochastically dominates $\{t_j\}_{j=1}^{n-1}$ and $\{t_j^-\}_{j=1}^{n-1}$ is stochastically dominated by $\{t_j\}_{j=1}^{n-1}$ in the sense that for any $x \in \mathbb{R}_+^{n-1}$, we have that
\begin{alignat*}{2}
\Pr[t_1^- >x_1  , \ldots , t_{n-1}^- > x_{n-1}]
&\leq \Pr[t_1>x_1  , \ldots , t_{n-1} > x_{n-1}] \\
&\leq \Pr[t_1^+>x_1  , \ldots , t_{n-1}^+ > x_{n-1}].
\end{alignat*}
\end{lemma}
\begin{proof}
We first prove the statement for $\{t_j^+\}_{j=1}^{n-1}$. The independence of $\{t_j^+\}_{j=1}^{n-1}$ implies that
\begin{equation}
\label{eq:probInd}
\Pr[t_1^+>x_1  , \ldots , t_{n-1}^+ > x_{n-1}]
= \prod_{j=1}^{n-1} \Pr[t_j^+ > x_j] = \prod_{j=1}^{n-1} e^{-f(n,j)x_j}.
\end{equation}
On the other hand we can estimate $\Pr[t_1>x_1  , \ldots , t_{n-1} > x_{n-1}]$ as follows. First note that by conditioning on $I_1$ we get for any $x_1 > 0$ that
\[
\begin{split}
&\Pr[t_1>x_1  , \ldots , t_{n-1} > x_{n-1}] \\
&\quad = \mathbb{E}\big[\Pr[t_2>x_2, \ldots , t_{n-1} > x_{n-1} ~\vert~ I_1,t_1>x_1]\Pr[t_1>x_1 ~\vert~ I_1]\big].
\end{split}
\]
Since $f(n,j) \leq \mathbb{E}[t_j ~ \vert~ I_j]^{-1}$ we obtain that this is at most
\[
e^{-f(n,1) x_1} P(t_2>x_2, \ldots , t_{n-1} > x_{n-1} ~\vert~ t_1>x_1).
\]
Moreover, for any $1 < k \le n-2$
we obtain 
\[
\begin{split}
& \Pr[t_k >x_k, \ldots , t_{n-1} > x_{n-1} ~\vert~ t_1>x_1 ,\ldots, t_{k-1} > x_{k-1}] \\
= &~\mathbb{E}\big[\Pr[t_{k+1} >x_{k+1}, \ldots , t_{n-1} > x_{n-1} ~\vert~ t_1>x_1 ,\ldots, t_k > x_k, I_k] \\
& \hspace{4cm} \cdot \Pr[t_k > x_k ~\vert~ t_1>x_1 ,\ldots, t_{k-1} > x_{k-1}, I_k]\big].
\end{split}
\]
By applying Lemma~\ref{Masterlemma} we infer that $t_k$ conditioned on any value of $I_k$ is independent of $t_1, \dots, t_{k-1}$. Therefore
\[
\Pr[t_k > x_k ~\vert~ t_1>x_1 ,\ldots, t_{k-1} > x_{k-1}, I_k] = \Pr[t_k > x_k ~\vert~ I_k],
\]
so that 
\[
\begin{split}
& \Pr[t_k >x_k, \ldots , t_{n-1} > x_{n-1} ~\vert~ t_1>x_1 ,\ldots, t_{k-1} > x_{k-1}] \\
\le & ~ e^{-f(n,k)x_k} \Pr[t_{k+1} >x_{k+1}, \ldots , t_{n-1} > x_{n-1} ~\vert~ t_1>x_1 ,\ldots, t_{k} > x_{k}].
\end{split}
\]
By induction we finally arrive at
\[
\Pr[t_1>x_1  , \ldots , t_{n-1} > x_{n-1}] \leq \prod_{j=1}^{n-1} e^{-f(n,j) x_j},
\]
which, together with~\eqref{eq:probInd}, proves the second inequality in the conclusion of the lemma. The first inequality is proven completely analogouesly by assuming $g(n,j) \geq \mathbb{E}[t_j ~ \vert~ I_j]^{-1}$ instead of $f(n,j) \leq \mathbb{E}[t_j ~ \vert~ I_j]^{-1}$ in the previous calculations. \qed
\end{proof}

Note that Corollary~\ref{dense:t_j} and Lemma~\ref{sparse:tj} guarantee that we can apply the previous lemma with 
$$f(n,j), g(n,j) = \Theta(\min\{j, n-j\}),$$
uniformly for $1 \le j < n$. The second lemma estimates the moment-generating function of $T(G)$ with this assumption; this will be our main tool in studying the distribution of $T(G)$.
\begin{lemma}
\label{momentgen}
Let $\{t_j^+\}_{j=1}^{n-1}$ and $\{t_j^-\}_{j=1}^{n-1}$ be as in Lemma \ref{stochbound} with $f(n,j), g(n,j) = \Theta(\min\{j, n-j\})$. Moreover, let $T^+(G):=\sum_{j=1}^{n-1} t_j^+$ and $T^-(G):=\sum_{j=1}^{n-1} t_j^-$. Then
\begin{alignat*}{2}
	\textrm{for } 0 < \lambda < \min_{j\in[n-1]} f(n,j), & \quad \mathbb{E}\left[e^{\lambda T(G)} \right] &\leq \exp\left\{\lambda \mathbb{E}[T^+(G)]+O(1)\right\}, \\
	\textrm{for } \lambda <0, & \quad \mathbb{E}\left[e^{\lambda T(G)} \right] &\leq \exp\left\{\lambda \mathbb{E}[T^-(G)]+O(1) \right\}.
\end{alignat*}
\end{lemma}
\begin{proof}
We first prove the statement for $\lambda >0$. The fact that $\{t_j^+\}_{j=1}^{n-1}$ stochastically dominates $\{t_j\}_{j=1}^{n-1}$ implies
\[
\mathbb{E}\left[e^{\lambda T(G)}\right] \leq \mathbb{E}\left[e^{\lambda T^+(G)}\right]
\]
Using that $\{t_j^+\}_{j=1}^{n-1}$ are independent we obtain
\[
\mathbb{E}\left[e^{\lambda T^+(G)}\right] = \prod_{j=1}^{n-1} \mathbb{E}\left[e^{\lambda t_j^+}\right] =\prod_{j=1}^{n-1}\frac{1}{1-\lambda f(n,j)^{-1}}.
\]
Here we have to restrict $\lambda$ to the interval $(0,\min_{j\in[n-1]} f(n,j))$, as otherwise some of the moment-generating functions might not exist. Since 
$$f(n,j)= \Theta(\min\{j, n-j\})$$ 
we infer that
\begin{align*}
\frac{1}{1-\lambda f(n,j)^{-1}} &= \sum_{k=0}^{\infty} (\lambda f(n,j)^{-1})^k \\ 
&= 1 + \lambda f(n,j)^{-1} + O \left(\min\{j^2, (n-j)^2\}^{-1} \right).
\end{align*}
Thus, so far we obtained that
\[
\mathbb{E}\left[e^{\lambda T(G)}\right]
\leq \prod_{j=1}^{n-1} \left( 1 + \lambda f(n,j)^{-1} + O \Big(\min\{j^2, (n-j)^2\}^{-1}  \Big)\right).
\]
Taking logarithms yields
\[
\ln \left(\mathbb{E}\left[e^{\lambda T(G)}\right] \right) \leq \sum_{j=1}^{n-1} \ln\left( 1 + \lambda f(n,j)^{-1} + O \Big(\min\{j^2, (n-j)^2\}^{-1} \Big)\right) .
\]
Using the estimate $\ln(1+x) \le x$ , $x\in(-1,1)$ we get that
\[
\begin{split}
\ln \left(\mathbb{E}\left[e^{\lambda T(G)}\right] \right)
&\leq \sum_{j=1}^{n-1} \lambda f(n,j)^{-1} + O \left(\min\{j^2, (n-j)^2\}^{-1}  \right)  \\
&= \lambda \mathbb{E}[T^+(G)] + O\left(1\right),
\end{split}
\]
which proves the first statement. The second statement follows completely analogously by using the fact $\mathbb{E}[e^{\lambda T(G)}] \leq \mathbb{E}[e^{\lambda T^-(G)}]$ for $\lambda < 0$ and that the moment-generating function $\mathbb{E}[e^{\lambda X}]$ of an exponentially distributed random variable $X$ exists for any $\lambda < 0$. \qed
\end{proof}

We move on to the two Corollaries \ref{dense: T(G)} and \ref{sparse: T(G)} that conclude the proofs of Theorem \ref{MainThm1} and \ref{MainThm2}. 

\begin{corollary}
\label{dense: T(G)}
Let $G$ be any graph with $n$ nodes that satisfies the conclusion of Lemma~\ref{lemma:dense}.  Then for any $\lambda \in (0,2)$  we have for $n$ large enough and any $t>0$ that
\[
\Pr\left[ \vert T(G) - \mathbb{E}[T(G)] \vert > t \right] \leq \exp\left\{ 3\sqrt{\frac{34}{\alpha(n)}}\ln(n) - \lambda t + O\left(1\right)\right\}.
\]
\end{corollary}
\begin{proof}
We prove the statement separately for the upper tail and the lower tail. We begin with the upper tail. Here, Markov's inequality for the monotonically increasing function $e^{\lambda x}$, $\lambda >0$ implies
\begin{equation}
\label{eq:Tuppera}
\Pr\left[ T(G) > \mathbb{E}[T(G)] + t \right] \leq \mathbb{E}\left[e^{\lambda T(G)} \right] e^{-\lambda \mathbb{E}[T(G)]- \lambda t}.
\end{equation}
We use that a sequence $\{t_j^+\}_{j=1}^{n-1}$ of independent random variables, where~$t_j^+ \sim \mathsf{Exp}\left(\left(1-\sqrt{33/\alpha(n)} \right) 2j(n-j)/(n-1) \right)$, dominates stochastically $\{t_j\}_{j=1}^{n-1}$. This can be inferred in applying Corollary \ref{dense:t_j} to Lemma \ref{stochbound} above. Letting $T^+(G):=\sum_{j=1}^{n-1} t_j^+$ we can thus bound the above expression by
\begin{equation}
\label{eq:boundE[exp(lambdaT(G)]}
\mathbb{E}\left[e^{\lambda T(G)} \right] \le \mathbb{E}\big[e^{\lambda T^+(G)} \big]
\le \mathbb{E}\big[e^{\lambda T^+(G) + O(1)} \big]
= \exp\left\{\lambda \mathbb{E}[T^+(G)] + O(1)\right\}.
\end{equation}
For the second inequality in Eq.~\eqref{eq:boundE[exp(lambdaT(G)]} we use Lemma~\ref{momentgen}. Note that $\lambda$ has to be restricted to the interval $\big(0,2\big(1-\sqrt{33/\alpha(n)} \big)\big)$ here. Plugging this into~\eqref{eq:Tuppera} and using that $\lambda < 2$ and $H_n = \ln(n) + O(1)$ we arrive at the bound
\begin{equation}
\label{eq: upper bound}
\Pr\left[ T(G) > \mathbb{E}[T(G)] + t \right] \leq 
\exp\left\{ 4\sqrt{\frac{34}{\alpha(n)}}\ln(n) - \lambda t + O(1)\right\}.
\end{equation}

For the lower tail we also apply Markov's inequality for $e^{\lambda x}$, $\lambda >0$ and obtain
$\Pr \left( T(G) < \mathbb{E}[T(G)] - t \right) \leq \mathbb{E}\left[e^{-\lambda T(G)} \right] \exp\left\{ \lambda \mathbb{E}[T(G)] - \lambda t \right\}$.
To bound this expression we use that a sequence $\{t_j^-\}_{j=1}^{n-1}$ of independent random variables, where $t_j^- \sim \mathsf{Exp}\big(\big(1+\sqrt{33/\alpha(n)} \big)2j(n-j)/(n-1)\big)$, is stochastically dominated by $\{t_j\}_{j=1}^{n-1}$. Again, this can be inferred by applying Corollary \ref{dense:t_j} to Lemma \ref{stochbound}. After the same steps as for the upper tail we deduce that for $\lambda < 2$ we get the same bound as in \eqref{eq: upper bound}. \qed
\end{proof}

\begin{corollary}
\label{sparse: T(G)}
Let $G$ be any graph with $n$ nodes that satisfies the conclusions of Lemma \ref{lemma:deg(v)} -- \ref{lemma:N(v)} and take $C'>0$ from Corollary~\ref{sparse:tj}. Then for any $\lambda \in (0,C')$ we have for $n$ large enough and any $t>0$
\[
\Pr \left[ T(G) - \mathbb{E}[T(G)]  > t \right] \leq \exp\left\{- \lambda t + O\left(\ln^{3/4}(n)\right)\right\},
\]
\end{corollary}

\begin{proof}
The proof is similar to the proof of Corollary~\ref{dense: T(G)}, so we highlight only the differences. Applying Markov's inequality implies that 
\begin{equation}
\label{eq:Tuppera2}
\Pr[T(G) > \mathbb{E}[T(G)] + t ]\leq \mathbb{E}\left[e^{-\lambda T(G)} \right] \exp\left\{ - \lambda \mathbb{E}[T(G)] + \lambda t \right\}.
\end{equation}
Applying Lemma \ref{stochbound} and Corollary \ref{dense:t_j} we infer that a sequence $\{t_j^+\}_{j=1}^{n-1}$ of independent random variables stochastically bounds $\{t_j\}_{j=1}^{n-1}$ if we choose $t_j^+$ to be exponentially distributed with parameter 
\begin{alignat*}{2}
C\min\{j,n-j\}, ~\big(1 - C\ln^{-1/4}(n)\big)2j\textrm{, or}~\big(1 - C\ln^{-1/4}(n)\big)(n-j).
\end{alignat*}
according to i) -- v) of Corollary \ref{sparse:tj}. Moreover letting $T^+(G):=\sum_{j=1}^{n-1} t_j^+$ we can determine the expectation of $T^+(G)$ in repeating the calculation in the proof of Lemma \ref{sparse: expectation} and get
\[
\mathbb{E}[T^+(G)]=\ln(n) + O\left(\ln^{3/4}(n)\right).
\]
By applying Lemma \ref{momentgen} and using the above considerations we can bound \eqref{eq:Tuppera2} by
\[
\exp\left\{- \lambda t + O\left(\ln^{3/4}(n)\right)\right\},
\]
where $\lambda$ has to be restricted to the interval $(0,C)$.\qed
\end{proof}

\section{Variations of the Asynchronous Push-Pull Protocol}
\label{sec:vop}
\subsection{The Effect of Transmission Failures -- Proposition~\ref{MainThm3}}

In this section we consider a more general version of the asynchronous push-pull protocol, in which nodes succeed to push or pull the rumor with probability $q \in (0,1]$ and fail to do so with probability $1-q$ independently of any other contacts established between any two nodes. Our aim is to give a statement in the spirit of Lemma \ref{Masterlemma} for this ``faulty'' version. We get the following quantitative statement.
\begin{lemma}
\label{Masterlemma2}
Let $1 \le j < n$. Then $t_j$ is exponentially distributed with parameter $q \cdot Q_j$, where $Q_j$ is given in Lemma \ref{Masterlemma}.
Moreover, conditional on $I_j$ the time $t_j$ is independent of $t_1, \ldots t_{j-1}$.
\end{lemma}
\begin{proof}
The probability that $v \in I_j$ informs a uninformed node in a push attempt is~$q{\vert N(v) \cap U_j \vert}/{d(v)}$. Similarly, the probability that $w\in U_j$ is informed in a pull attempt is~$q{\vert N(w) \cap I_j \vert}/{d(v)}$. Therefore, the probability that a node is informed in a push or pull attempt is 
$$
q_j(q) = \frac{q}{n}~\left(\sum_{v \in I_j} \vert N(v) \cap U_j \vert/d(v) + \sum_{w \in U_j} \vert N(w) \cap I_j \vert/d(w) \right).
$$ 
As in the proof of Lemma \ref{Masterlemma}, it follows that $t_j$ is exponentially distributed with parameter $n q_j(q) = q \cdot Q_j$. \qed
\end{proof}
An immediate consequence of Lemma~\ref{Masterlemma2} is that for arbitrary graphs $G$, $\mathbb{E}[T_q(G)] = 1/q~\mathbb{E}[T(G)]$. Moreover, we can repeat all steps performed in Sections~\ref{sec:pthm1thm2} and~\ref{sec:pt3t4} to study the effect of the success probability $q$ on the time that is required to spread the rumor to all nodes of $G_{n,p}$. The steps are literally the same, with the only difference being that the parameters in all involved exponentially distributions are multiplied by an additional factor of $q$. For example, in Corollary~\ref{dense:t_j}, Equation~\eqref{eq:parboundsdense} is replaced by $\left(1 \pm \sqrt{33/\alpha(n)} \right) 2qj(n-j)/(n-1) $
and similarly, in the conclusions i)--v) of Corollary~\ref{sparse:tj} we get the bounds $\Theta(q \min\{j, n-j\})$, $(1 - O(\ln^{-1/4}(n)))2qj$, and $(1 - O(\ln^{-1/4}(n)))q(n-j)$. The proofs in Section~\ref{sec:pt3t4} are adapted accordingly to obtain asymptotically the same bounds as in Theorems~\ref{MainThm1} and~\ref{MainThm2}.

\subsection{The Effect of Faulty Nodes -- Proposition~\ref{MainThm4}}

Suppose that before the rumor is spread by the asynchronous push-pull protocol, a random subset of the nodes of $G_{n,p}$ is declared ``faulty'', in the sense that even if they receive the rumor, they will neither perform any push operation, nor will they respond to any pull request.

Note that if the subset of faulty nodes is of size $o(n)$ and $p \ge (1 + \eps)\ln(n)/n$ for some $\eps>0$, then the subgraph of $G_{n,p}$ induced by the non-faulty nodes is distributed like $G_{n',p'}$, where
\begin{equation}
\label{eq:faulty nodes}
n' = (1-o(1))n \qquad \text{ and } \qquad p' \ge (1+\eps - o(1))\frac{\ln(n)}{n}.
\end{equation}
Thus, if the initially informed node does not fault, then the results in the previous sections apply also in this case, where we replace $n$ and $p$ by the values given in~\eqref{eq:faulty nodes}.

\section{Numerical simulations}
\label{sec:numsim}

\begin{figure}
\centering
\includegraphics[width=\textwidth]{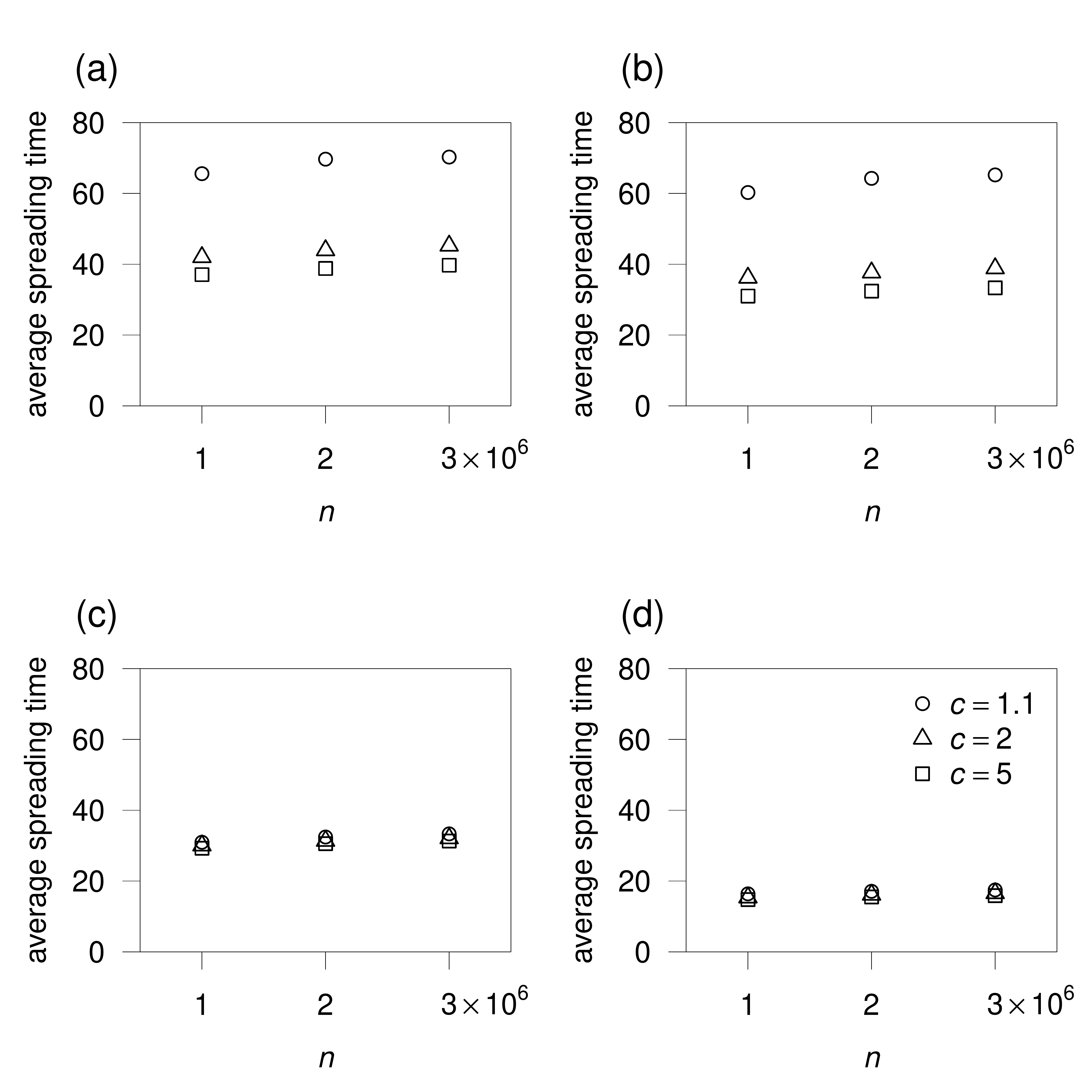}
\caption{
Spreading time of (a) synchronous push, (b) asynchronous push, (c) asynchronous pull, and (d) asynchronous push-pull averaged over 1000 realizations. Standard errors of the means are smaller than the plot symbols.  
\label{fig:sim}
}
\end{figure}

We performed numerical simulations of push, pull, and push-pull to compare the dependency of spreading times with respect to the edge probability $p$. For each protocol, we averaged the spreading time over 1000 realizations, where a realization involves generating the Erd\H{o}s-R\'{e}nyi random graph and simulating the protocol once on this graph. Random graphs where generated for each combination of $n \in \{ 10^6,2\cdot10^6,3\cdot10^6 \}$, and $p = c\ln(n)/n$, where $c \in \{ 1.1, 2, 5\}$. The results are shown in Figure~\ref{fig:sim}. 

We find that the spreading time of push significantly increases as $c$ decreases, while the spreading times of asynchronous pull and push-pull remain largely unaffected. In particular, dependency on $c$ of the spreading time of push is similar in the synchronous and asynchronous version. These results suggest that the $c$-independent bound on the spreading time of asynchronous push-pull (Theorem~\ref{MainThm2}) is obtained with the help of pull. In fact, this is expected because we can proof a $c$-independent bound of $2\ln(n) + O(\ln^{3/4}(n))$ for asynchronous pull. To obtain this bound, we note that we can prove an analogue of Corollary 2 with the difference being that (ii) is replaced by $\big(1 \pm C\ln^{-1/4}(n)\big)j$ for $\ln(n) \le j\le n/\ln^3(n)$. Then, by repeating steps in Lemma~\ref{sparse: expectation} and Corollary~\ref{sparse: T(G)}, we obtain the claimed bound.

Interestingly, push-pull is almost twice as fast as pull, suggesting that push significantly contributes to the spreading time.  For instance, for $n=3\cdot10^6$ in Fig.~\ref{fig:sim}, push-pull is $1.91$, $1.94$, and $1.97$ times faster than pull for $c = 1.1, 2$, and $5$, respectively.

\bibliographystyle{spmpsci}
\bibliography{bibliography}{}

\begin{thebibliography}{10}
\providecommand{\url}[1]{{#1}}
\providecommand{\urlprefix}{URL }
\expandafter\ifx\csname urlstyle\endcsname\relax
  \providecommand{\doi}[1]{DOI~\discretionary{}{}{}#1}\else
  \providecommand{\doi}{DOI~\discretionary{}{}{}\begingroup
  \urlstyle{rm}\Url}\fi

\bibitem{Acan_arXiv}
Acan, H., Collevecchio, A., Mehrabian, A., Wormald, N.: On the push \& pull
  protocol for rumour spreading.
\newblock arXiv:1411.0948

\bibitem{AZ10}
Aigner, M., Ziegler, G.: Proofs from the Book.
\newblock Springer, Berlin Heidelberg (2010)

\bibitem{BA99}
Barab{\'a}si, A., Albert, R.: Emergence of scaling in random networks.
\newblock Science \textbf{286}, 509--512 (1999)

\bibitem{MR1864966}
Bollob{\'a}s, B.: Random Graphs, \emph{Cambridge Studies in Advanced
  Mathematics}, vol.~73, 2nd edn.
\newblock Cambridge University Press, Cambridge, United Kingdom (2001)

\bibitem{Bollobas1997book}
Bollob\'{a}s, B., Kohayakawa, Y.: On Richardson’s Model on the Hypercube.
\newblock Cambridge University Press, Cambridge, United Kingdom (1997)

\bibitem{BABD05}
Boyd, S., Arpita, G., Balaji, P., Devavrat, S.: Gossip algorithms: Design,
  analysis and applications.
\newblock In: Proceedings of the 24th Annual Joint Conference of the IEEE
  Computer and Communications Societies (INFOCOM'05), pp. 1653--1664. Miami,
  FL, USA (2005)

\bibitem{Boyd06}
Boyd, S., Ghosh, A., Prabhakar, B., Shah, D.: Randomized gossip algorithms.
\newblock IEEE Transactions on Information Theory \textbf{52}, 2508--2530
  (2006)

\bibitem{CLP10-2}
Chierichetti, F., Lattanzi, S., Panconesi, A.: Almost tight bounds for rumour
  spreading with conductance.
\newblock In: Proceedings of the 42nd ACM Symposium on Theory of Computing
  (STOC '10), pp. 399--408. Cambridge, MA, USA (2010)

\bibitem{CLP10}
Chierichetti, F., Lattanzi, S., Panconesi, A.: Rumour spreading and graph
  conductance.
\newblock In: Proceedings of the 21st annual ACM-SIAM symposium on Discrete
  Algorithms (SODA'10), pp. 1657--1663. Austin, TX, USA (2010)

\bibitem{CL03}
Chung, F., Lu, L.: The average distance in a random graph with given expected
  degrees.
\newblock Proceedings of the National Academy of Sciences of the United States
  of America \textbf{99}, 15,879–--15,882 (2002)

\bibitem{CF07}
Cooper, C., Frieze, A.: The cover time of sparse random graphs.
\newblock Random Structures and Algorithms \textbf{30}, 1--16 (2007)

\bibitem{Dem87}
Demers, A., Greene, D., Hauser, C., Irish, W., Larson, J., Shenker, S.,
  Sturgis, H., Swinehart, D., Terry, D.: Epidemic algorithms for replicated
  database maintenance.
\newblock In: Proceedings of the 6th Annual ACM Symposium on Principles of
  Distributed Computing (POCD'87), pp. 1--12. Vancouver, BC, Canada (1987)

\bibitem{DFF11}
Doerr, B., Fouz, M., Friedrich, T.: Social networks spread rumors in
  sublogarithmic time.
\newblock Electronic Notes in Discrete Mathematics \textbf{38}, 303--308 (2011)

\bibitem{DFF12}
Doerr, B., Fouz, M., Friedrich, T.: Asynchronous rumor spreading in
  preferential attachment graphs.
\newblock In: Proceedings of the 13th Scandinavian Workshop on Algorithm Theory
  (SWAT'12), pp. 307--315. Helsinki, Finland (2012)

\bibitem{feller_book}
Feller, W.: An Introduction to Probability Theory and Its Applications, vol.~1.
\newblock John Wiley \& Sons, New York, NY, USA (1968)

\bibitem{Fill1993}
Fill, J.A., Pemantle, R.: Percolation, first-passage percolation and covering
  times for richardson's model on the $n$-cube.
\newblock The Annals of Applied Probability \textbf{3}, 593--629 (1993)

\bibitem{FHP10}
Fountoulakis, N., Huber, A., Panagiotou, K.: Reliable broadcasting in random
  networks and the effect of density.
\newblock In: Proceedings of the 29th Conference on Computer Communications
  (INFOCOM '10), pp. 2552--2560. San Diego, CA, USA (2010)

\bibitem{FPS12}
Fountoulakis, N., Panagiotou, K., Sauerwald, T.: Ultra-fast rumor spreading in
  social networks.
\newblock In: Proceedings of the 23rd ACM-SIAM Symposium on Discrete Algorithms
  (SODA '12), pp. 1642--1660. Kyoto, Japan (2012)

\bibitem{FG85}
Frieze, A., Grimmett, G.: The shortest-path problem for graphs with random
  arc-lengths.
\newblock Discrete Applied Mathematics \textbf{10}, 57--77 (1985)

\bibitem{Gia11}
Giakkoupis, G.: Tight bounds for rumor spreading in graphs of a given
  conductance.
\newblock In: Proceedings of the 28th International Symposium on Theoretical
  Aspects of Computer Science (STACS'11), pp. 57--68. Dortmund, Germany (2011)

\bibitem{Janson1999}
Janson, S.: One, two and three times log n/n for paths in a complete graph with
  random weights.
\newblock Combinatorics, Probability \& Computing \textbf{8}, 347--361 (1999)

\bibitem{JVGKS07}
Jelasity, M., Voulgaris, S., Guerraoui, R., Kermarrec, A.M., van Steen~M.:
  Gossip-based peer sampling.
\newblock ACM Transactions on Computer Systems \textbf{25}, 8 (2007)

\bibitem{KSSV00}
Karp, R., Schindelhauer, C., Shenker, S., V{\"o}cking, B.: Randomized rumor
  spreading.
\newblock In: Proceedings of the 41th Annual Symposium on Foundations of
  Computer Science (FOCS'00), pp. 565--574. Redondo Beach, CA, USA (2000)

\bibitem{Panagiotou2013}
Panagiotou, K., Fountoulakis, N.: Rumor spreading on random regular graphs and
  expanders.
\newblock Random Structures and Algorithms \textbf{43}, 201--220 (2013)

\bibitem{PPSS}
Panagiotou, K., P\'erez-Gim\'enez, X., Sauerwald, T., Sun, H.: Randomized
  rumour spreading: The effect of the network topology.
\newblock Combinatorics, Probability and Computing \textbf{24}, 457--479 (2015)

\bibitem{PS2013}
Panagiotou, K., Speidel, L.: Asynchronous rumor spreading on random graphs.
\newblock In: Proceedings of the 24th International Symposium on Algorithms and
  Computation (ISAAC'13), pp. 424--434. Hong Kong, China (2013)

\bibitem{Pit87}
Pittel, B.: On spreading a rumor.
\newblock SIAM Journal on Applied Mathematics \textbf{47}, 213--223 (1987)

\bibitem{RMH98}
van Renesse, R., Minsky, Y., Hayden, M.: A gossip-style failure detection
  service.
\newblock In: Proceedings of the IFIP International Conference on Distributed
  Systems Platforms and Open Distributed Processing (Middleware'98), pp.
  55--70. The Lake District, United Kingdom (1998)

\end{thebibliography}
\end{document}